\documentclass[12pt]{article}
\usepackage[a4paper,top=2cm,bottom=2cm,left=2cm,right=2cm]{geometry}
\usepackage{amsmath}
\usepackage{amsthm}
\usepackage{amssymb}

\usepackage{subfig}
\usepackage[most]{tcolorbox}
\newtcbtheorem{mytcbprob}{Problem}{}{problem}
\usepackage{graphicx}
\usepackage{braket}
\usepackage[colorlinks=true, allcolors=blue]{hyperref}
\usepackage{listings}
\usepackage[ruled]{algorithm2e}
\usepackage{enumerate}
\usepackage{url}
\usepackage{pgfplots}
\pgfplotsset{compat=1.15}
\usetikzlibrary{arrows}
\usepackage[title]{appendix}
\usepackage{graphicx}

\newtheorem{lemma}{Lemma}
\newtheorem{theorem}{Theorem}

\title{Optimal exact quantum algorithm for the promised element distinctness problem }
\author{Guanzhong Li\thanks{Email: ligzh9@mail2.sysu.edu.cn}~ and Lvzhou Li\thanks{Email: lilvzh@mail.sysu.edu.cn (corresponding author)}\\
\small{{\it Institute of Quantum Computing and Computer Theory,}}\\
\small{{\it School of Computer Science and Engineering,}}\\
\small {{\it  Sun Yat-sen University, Guangzhou 510006, China}}}

\begin{document}
\maketitle

\begin{abstract}	
The element distinctness problem is to determine whether a string $x=(x_1,\ldots,x_N)$ of $N$ elements contains two elements of the same value (a.k.a colliding pair), for which Ambainis proposed an optimal quantum algorithm. The idea behind Ambainis' algorithm is to first reduce the problem to the promised version in which $x$ is promised to contain at most one colliding pair, and then design an algorithm $\mathcal{A}$ requiring $O(N^{2/3})$ queries based on quantum walk search for the promise problem. However, $\mathcal{A}$ is probabilistic and may fail to give the right answer. We thus, in this work, design an exact quantum algorithm for the promise problem which never errs and requires $O(N^{2/3})$ queries. This algorithm is  proved optimal. Technically, we modify the quantum walk search operator on quasi-Johnson graph to have arbitrary phases, and then use Jordan's lemma as the analyzing tool to reduce the quantum walk search operator to the generalized Grover's operator. This allows us to utilize the recently proposed fixed-axis-rotation (FXR) method for exact quantum search, and hence achieve 100\% success.
\end{abstract}

\section{Introduction}
In the element distinctness problem, we are given a black box (oracle) $O_x$ that when queried index $i$ of the unknown string $x=(x_1,x_2,\ldots,x_N)\in [M]^N$, where $[M]=\{1,2,\ldots,M\}$ and $M\geq N$, it  outputs the value $x_i$.
The task now is to determine whether the string $x$ contains two equal items $x_{i}=x_{j}$ (called the colliding pair), with as few queries to the index oracle $O_x$ as possible.
The problem is well-known both in classical and quantum computation. The classical bounded-error query complexity (a.k.a decision tree complexity, see \cite{decision_tree_survey} for more information) is shown to be $\Theta(N)$ by a trivial reduction from the OR-problem \cite{quantum_alg1.2}. In comparison, Ambainis proposed a $O(N^{2/3})$-query quantum algorithm with bounded-error \cite{quantum_alg2.1,quantum_alg2.2}.
The algorithm is also optimal: the bounded-error quantum query complexity of the element distinctness problem has previously been shown to be $\Omega(N^{2/3})$ \cite{quantum_lower1,quantum_lower2,quantum_lower3} by the polynomial method \cite{polynomial}.

In \cite{quantum_alg2.1,quantum_alg2.2}, Ambainis first designed a $O(N^{2/3})$-query quantum algorithm $\mathcal{A}$ which solves the promised version in which  $x$ contains at most one colliding pair, by using quantum walk search on Johnson graph. 
He then reduced the problem to this promised version, by sampling a sequence of exponentially shrinking subsets to run $\mathcal{A}$. 
By showing the probability that the sequence has a subset containing at most one colliding pair is greater than $1/2$ and the overall query complexity is of the same order, Ambainis successfully solved the element distinctness problem with query complexity $O(N^{2/3})$. 
This two-stage process shows the promise problem that the algorithm $\mathcal{A}$ solves is crucial, and we formalize it in the following:

\begin{mytcbprob}{element distinctness problem with at most one colliding pair }{at_most_1}
\textbf{input:} index oracle $O_x$ hiding the unknown string $x=(x_1,x_2,\ldots,x_N)\in [M]^N$.

\textbf{promise:} $x$ contains exactly one colliding pair $x_i=x_j$ with $i\neq j$; or $x_i$s are distinct for $i\in[N]$.

\textbf{output:} index $(i,j)$ of the colliding pair $x_i=x_j$; or ``all distinct”.
\end{mytcbprob}

In fact, $\mathcal{A}$ is a one-sided error algorithm: if the string $x$ contains a colliding pair $x_i=x_j$, then $\mathcal{A}$ might asserts that $x$ is ``all distinct''; but if elements in $x$ are all distinct, then $\mathcal{A}$ won't make a mistake. This inspires us to consider the question: is there an exact quantum algorithm that solves Problem~\ref{problem:at_most_1} with the same query complexity $O(N^{2/3})$? In this paper, we answer the question affirmatively by presenting Algorithm~\ref{alg:main_block} in Section~\ref{sec:alg}. We summarize our result as:

\begin{theorem}\label{thm:intro}
There exits an exact quantum algorithm that solves Problem~\ref{problem:at_most_1} with query complexity $O(N^{2/3})$.
\end{theorem}

The main steps of  Algorithm~\ref{alg:main_block} is as follows:
\begin{enumerate}
    \item Converting Problem~\ref{problem:at_most_1} to quantum walk search on a quasi-Johnson graph $G$ (using the staggered model \cite{staggered}). This is done in Section~\ref{subsec:johonson} similar to \cite{Portugal_2018,Portugal_book}, but we modify the quantum walk operator $u := U_B(\theta_2) U_A(\theta_1)$ and the operator $R_T(\alpha)$ which marks the target vertices by multiplying them by the phase $e^{i\alpha}$, so that the angle parameters $\theta_1,\theta_2,\alpha$ are no longer restricted to $\pi$. The flexibility in setting these angle parameters offers an opportunity to design an exact algorithm.
    
    \item Reducing quantum walk search on the whole graph $G$ to its $5$-dimensional invariant subspace $\mathcal{H}_0$, by classifying the vertices of $G$, i.e. $(r+1)$-sized subsets $(S,y)$ of string $x$'s indices $[N]$, into $5$ groups according to the $5$ different cases in which subset $(S,y)$ intersect with the colliding pair. This is presented in Section~\ref{sec:reduction}, where we also obtain the expression of $U_B(\theta_2)$ and $U_A(\theta_1)$ in $\mathcal{H}_0$.
    
    \item Reducing further the $5$-dimensional subspace of quantum walk search to the $2$-dimensional qubit subspace $\mathcal{H}_1 := \text{span}\{\ket{\psi_0},\ket{T}\}$. This part (shown in Section~\ref{sec:clock}) is the core of this paper with quite some technique, and we elaborate it in the following steps:
    \begin{enumerate}[(i)]
        \item Based on Jordan's Lemma \cite{Jordan_1875}, we decompose the quantum walk operator $u=U_B(\theta_2) U_A(\theta_1)$ to the direct sum of a scalar-multiplication in the $1$-dimensional invariant subspace spanned by initial state $\ket{\psi_0}$, and two 3D-rotations on two different Bloch spheres (two $2$-dimensional invariant subspaces).
        \item Based on the special rotation angles $\phi_1,\phi_2$ of the two 3D-rotations obtained in step (i), and the {\it intuition} of the minute hand (the bigger angle $\phi_2$) catching up with the hour hand ($\phi_1$) on a watch, we devise the equations [Eq.~\eqref{eq:inner_angle1}] that the iteration number $ct_2$ of $u=U_B(\theta_2) U_A(\theta_1)$ and the two fixed angles $\theta_1,\theta_2$ therein should satisfy, so that the effect of $u^{ct_2}$ on the two Bloch spheres both become the identity transformation. The overall effect of $u^{ct_2}$ in $\mathcal{H}_0$ thus becomes a phase rotation about the initial state: $e^{-i\beta \ket{\psi_0}\bra{\psi_0}}$, where $\beta$ is some fixed angle to be determined in the next step.
        \item With some detailed analysis (Lemma~\ref{thm:inner_exist}) on Eq.~\eqref{eq:inner_angle1}, we show that it posses solutions $\theta_1,\theta_2$ when $ct_2$ is set according to Eq.~\eqref{eq:inner_iter1}, so that the intuition in step (ii) really works.
    \end{enumerate}
    By step (ii), we have actually obtained the $2$-dimensional invariant subspace $\mathcal{H}_1 = \text{span}\{\ket{\psi_0},\ket{T}\}$ of the quantum walk search iteration $u^{ct_2}R_T(\alpha)=e^{-i\beta \ket{\psi_0}\bra{\psi_0}} e^{i\alpha \ket{T}\bra{T}} =: S_r(\beta) S_o(\alpha)$, which is the generalized Grover's iteration $G(\alpha,\beta)$.
    
    
    \item We use the fixed-axis-rotation (FXR) method \cite{FXR}, in Section~\ref{sec:FXR}, to obtain the appropriate parameters $\alpha_1,\alpha_2$ of rotation $F_{xr,\beta}:=G(\alpha_2,\beta) G(\alpha_1,\beta)$ and its iteration number $t_1$ [see Eqs.~\eqref{eq:outer_iter1}-\eqref{eq:im1}], such that applying $F_{xr,\beta}^{t_1}$ to the initial state $\ket{\psi_0}$ produces precisely the target state $\ket{T}$.
\end{enumerate}

Problem \ref{problem:at_most_1} has also been studied by Portugal in \cite{Portugal_2018,Portugal_book}, where he improves the success probability of Algorithm $\mathcal{A}$ from constant to $1-O(N^{1/3})$, quickly approaching $1$, by careful analysis on the asymptotic behaviours of eigenvalues and eigenvectors of the quantum walk iterations $[(U_B(\pi) U_A(\pi))^{t_2}R_T(\pi)]^{t_1}$ when $N\to\infty$, and by setting the optimal iteration numbers $t_1,t_2$. 
However, asymptotic behaviours are not enough to design an exact algorithm. Therefore in step 3 to determine the effect of $u^{ct_2}$, we use Jordan's Lemma \cite{Jordan_1875} about common invariant subspaces of two reflections instead. 
Note that in our case we consider `reflections' with partial phase $\theta\neq \pi$. Jordan's Lemma has been widely used in the literature, such as the spectral lemma of quantum walk operator based on Markov chain \cite{Szegedy_03}, spectral decomposition of interpolated quantum walk operator \cite{Krovi2016}, and the qubitization technique in quantum singular value transformation \cite{QSVT}.

The promise in Problem~\ref{problem:at_most_1} that there is at most one colliding pair is in some sense necessary to design a fast and exact quantum algorithm, because without any promise on the number of equal elements in $x$, the deterministic quantum query complexity will be $\Omega(N)$. In fact, we can first reduce the problem of calculating the Boolean function $OR(y),\ y\in\{0,1\}^N$ to the element distinctness problem of a specific string $x\in \{0,1,\ldots,N\}^{N+1}$ such that $x_i=(1-y_i)i,\forall i\in[N]$ and $x_{N+1}=0$ \cite{quantum_alg1.2}. Then $OR(y)=1$ if and only if $x$ contains a colliding pair. This reduction tells us that the query complexity of the element distinctness problem must be greater than that of calculating the Boolean function $OR$. Secondly, the deterministic quantum query complexity of calculating $OR$ on the domain $\{0,1\}^N$ is shown to be $N$ by the polynomial method \cite{polynomial}. Therefore, any exact quantum algorithm solving the element distinctness problem without any promise on the number of colliding pairs must query the index oracle $O_x$ for $\Omega(N)$ times.

On the other hand, the deterministic quantum query complexity of Problem~\ref{problem:at_most_1} is $\Omega(N^{2/3})$, so that the $O(N^{2/3})$ in Theorem \ref{thm:intro} is actually $\Theta(N^{2/3})$!
Note, however, that the aforementioned $\Omega(N^{2/3})$ bounded-error quantum query complexity of the general element distinctness problem by the polynomial method {\it does not} imply the $\Omega(N^{2/3})$ lower bound of Problem~\ref{problem:at_most_1}, because the promised version is easier than the general version and thus the query lower bound smaller.
Nevertheless, we can use the adversary method instead, because the negative-weight adversary lower bound \cite{hoyer2007negative} of the decision version of Problem~\ref{problem:at_most_1} has been proved to be $\Omega(N^{2/3})$ \cite{belovs2012adversary,rosmanis2014quantum}. For completeness, we describe the proof of the $\Omega(N^{2/3})$ lower bound of Problem~\ref{problem:at_most_1} in Section \ref{sec:lower_bound}.

\section{\label{sec:alg}Description of the algorithm}
In this section, we present an exact quantum algorithm for Problem~\ref{problem:at_most_1} (see Algorithm~\ref{alg:main_block}).
\subsection{\label{subsec:johonson}Converting to quantum walk search on quasi-Johnson graph}

The quasi-Johnson graph $G$ underlying Algorithm~\ref{alg:main_block} has {\it vertex set} $V=\{(S,y):S\in S_r,\ y\in [N]\setminus S \}$, where  $S_r$ consists of all the $r$-sized subsets of $ [N]$. Index $y$ can be seen as the index to be swapped into $S$ in one step of quantum walk. $r$ is set to be
\begin{equation}\label{eq:r_def}
    r=\lfloor N^{{2}/{3}} \rfloor.
\end{equation}
An {\it edge} in $G$ connects two different vertices $(S,y)$ and $(S',y')$ if and only if $S=S'$ or $S\cup\{y\}=S'\cup\{y'\}$. 
Note that $|V|=\binom{N}{r}(N-r)$, which is not equal to $\binom{N}{r+1}$. Thus $G$ is not a real Johnson graph, whose edges connect two vertices, i.e. $(r+1)$-sized subsets, if and only if they differ in exactly one element.

Every vertex is equipped with $r+1$ slots of size $M$ to store the item values in string $x$ corresponding to the index set $S\cup\{y\}$. Thus the basis of the two working registers is of the following form:
\begin{equation*}
    \ket{S,y}\otimes \ket{x'_1,x'_2,\ldots,x'_{r+1}}.
\end{equation*}
Suppose we store the value $x'_i$ with $m$ qubits ($2^m\geq M$). Then the effect of the index oracle $O_x$ is
\begin{equation}\label{eq:oracle}
    O_x \ket{i}\ket{b}=\ket{i}\ket{b\oplus x_i},
\end{equation}
where $i\in [N]$ is the querying index, and $\oplus$ is bit-wise XOR on $m$ qubits. Thus $O^2=I$.

As vertex basis $\ket{S,y}\in \mathbb{C}^{\binom{N}{r}}\otimes \mathbb{C}^{ N-r } $ encodes the index set $S\cup\{y\}$, we will need a quantum subroutine $\mathcal{P}$ to obtain the ascending index sequence $(i_1,i_2,\ldots,i_r)$ corresponding to $S$, and the true index $\bar{y}\in [N]$ corresponding to $y$, so as to query the index oracle $O_x$. To be more precise, 
\begin{equation}\label{eq:SubroutineP}
    \mathcal{P}: \ket{S,y} \otimes \ket{b_1} \cdots \ket{b_r} \ket{b'} \mapsto \ket{S,y} \otimes \ket{b_1\oplus i_1} \cdots \ket{b_r\oplus i_r} \ket{b'\oplus \bar{y} }.
\end{equation}
Note that $\mathcal{P}^2=I$. The decoding procedure $\mathcal{P}$ might take a lot of space, but as we are only considering query complexity, this is not a problem.

The {\it edge set} of graph $G$ is actually determined by the two local diffusion operator $U_A(\theta_1)$ and $U_B(\theta_2)$, as the edges indicate the transiting restriction among vertices in a step of quantum walk $u=U_B(\theta_2) U_A(\theta_1)$. Operator $U_A(\theta_1)$ is defined as
\begin{equation}\label{eq:U_A_def}
    U_A(\theta_1) = I-(1-e^{i\theta_1}) \sum_{S\in S_r} \ket{A_S}\bra{A_S},
\end{equation}
where 
\begin{equation}\label{eq:As}
    \ket{A_S} = \frac{1}{\sqrt{N-r}} \sum_{y\in[N]\setminus y} \ket{S,y}
\end{equation}
is an equal-superposition of the clique $A_S=\{(S,y)\in V : y\in [N]\setminus S \}$ determined by $S$. Note that $|A_s|=N-r$ and there are $\binom{N}{r}$ cliques of this kind. From the definition of $U_A(\theta_1)$, we know that it transits a vertex to a superposition of vertices in the same clique, but never to vertex in other clique. Therefore, $U_A(\theta_1)$ corresponds to applying Grover diffusion operator in every clique $A_S$.

The other diffusion operator $U_B(\theta_2)$ is the combined effect of steps 3.(b).(ii)-(iv) in Algorithm~\ref{alg:main_block} on the first register. The operator $U_B^{\text{EXT}}(\theta_2)$ in step 3.(b).(iii) is a bit complicated, and is defined as
\begin{equation}\label{eq:U_B_def}
    U_B^{\text{EXT}}(\theta_2) = I- (1-e^{i\theta_2}) \sum_{(x_1',\ldots,x_{r+1}')} \sum_{[S,y]} \ket{B_{[S,y]}^{(x_1',\ldots,x_{r+1}')}} \bra{B_{[S,y]}^{(x_1',\ldots,x_{r+1}')}},
\end{equation}
where the first sum is over all the possible string $(x_1',\ldots,x_{r+1}')\in [M]^{r+1}$, and the second sum is over all the representatives $[S,y]$ of clique $B_{[S,y]}=\{ (S',y') : \{y'\}\cup S'=\{y\}\cup S \}$. Note that $|B_{[S,y]}|=r+1$ and there are $\binom{N}{r+1}$ of this kind.
The state $\ket{B_{[S,y]}^{(x_1',\ldots,x_{r+1}')}}$ is determined by $[S,y]$ and $(x_1',\ldots,x_{r+1}')$, and is defined as 
\begin{equation}
    \ket{B_{[S,y]}^{x_1',\ldots,x_{r+1}'}} = \frac{1}{\sqrt{r+1}} \sum_{y'\in S\cup \{y\}} \ket{S\cup\{y\}\setminus\{y'\},y'}\ket{\pi(x_1'),\ldots,\pi(x_{r+1}')},
\end{equation}
where the permutation $\pi$ on the second register is related to the permutation of indices in the first register. This relation is best illustrated by an example: suppose $[S,y]=[\{1,4\},3]$, then 
\begin{equation}\label{eq:example}
    \ket{B_{[S,y]}^{(x_1',x_2',x_3')}} = \frac{1}{\sqrt{3}} \Big( \ket{\{1,4\},3}\ket{x_1',x_2',x_3'} + \ket{\{3,4\},1}\ket{x_3',x_2',x_1'} + \ket{\{1,3\},4}\ket{x_1',x_3',x_2'} \Big).
\end{equation}
Compared to the first term in Eq.~\eqref{eq:example}, the third term exchanged the index $4$ and $3$. Therefore $\pi$ exchanges the corresponding $x_2'$ and $x_3'$ in the second register. The second term is obtained in a similar fashion. 
Although the indices in $S$ are unordered, we regard them to be in the ascending order so that $\ket{B_{[S,y]}^{x_1',\ldots,x_{r+1}'}}$ is well-defined. For example, $x_3'$ in the second term of Eq.~\eqref{eq:example} is related to index $3$ but not $4$. 
Combining this and the ascending order of the output sequence of $\mathcal{P}$, we know that during the whole procedure of Algorithm~\ref{alg:main_block}, $\ket{x_{i_1},\ldots,x_{i_r}}$ in the second register as the subset of $x=(x_1,\ldots,x_N)$ is always related to the ascending index sequence $(i_1,\ldots,i_r)$ of $\ket{S}$ in the first register. 
Thus we can omit the second register when analyzing the algorithm. Moreover, in the final measurement of Algorithm~\ref{alg:main_block}, the two registers will collapse to the entangled state $\ket{S,y}\ket{x_{i_1},\ldots,x_{i_r}}$. Therefore, we can obtain indices of the colliding pair without further quires.
Thus we have shown, the combined effect of steps 3.(b).(ii)-(iv) in Algorithm~\ref{alg:main_block} can be seen as the following operator acting on the first register.
\begin{equation}
    U_B(\theta_2) = I- (1-e^{i\theta_2}) \sum_{[S,y]} \ket{B_{[S,y]}} \bra{B_{[S,y]}},
\end{equation}
where
\begin{equation}
    \ket{B_{[S,y]}} = \frac{1}{\sqrt{r+1}} \sum_{y'\in S\cup \{y\}} \ket{S\cup\{y\}\setminus\{y'\},y'}.
\end{equation}

The {\it marked vertices} $(S,y)$ of graph $G$ are those containing the colliding pair's indices $(i_1,i_2)$ in $S$, and the related marking operator $R_T(\alpha)$ multiplies $\ket{S,y}\ket{x'_1,\ldots,x'_{r},0}$ by $e^{i\alpha}$ if the second register contains the colliding pair. Note that $R_T(\alpha)$ may be time consuming, but it does not need to query the index oracle $O_x$.

An example of the quasi-Johnson graph when $N=5,r=2$ and $(1,3)$ being the indices of the colliding pair in $x=(x_1,\ldots,x_5)$ is shown in Fig.~\ref{fig:johnson}. 
It gives us a visual impression of what the quasi-Johnson graph $G$ with its two kinds of cliques $A_S$ and $B_{[S,y]}$ would look like.

\begin{figure}
    \centering
    \includegraphics[width=0.5\textwidth]{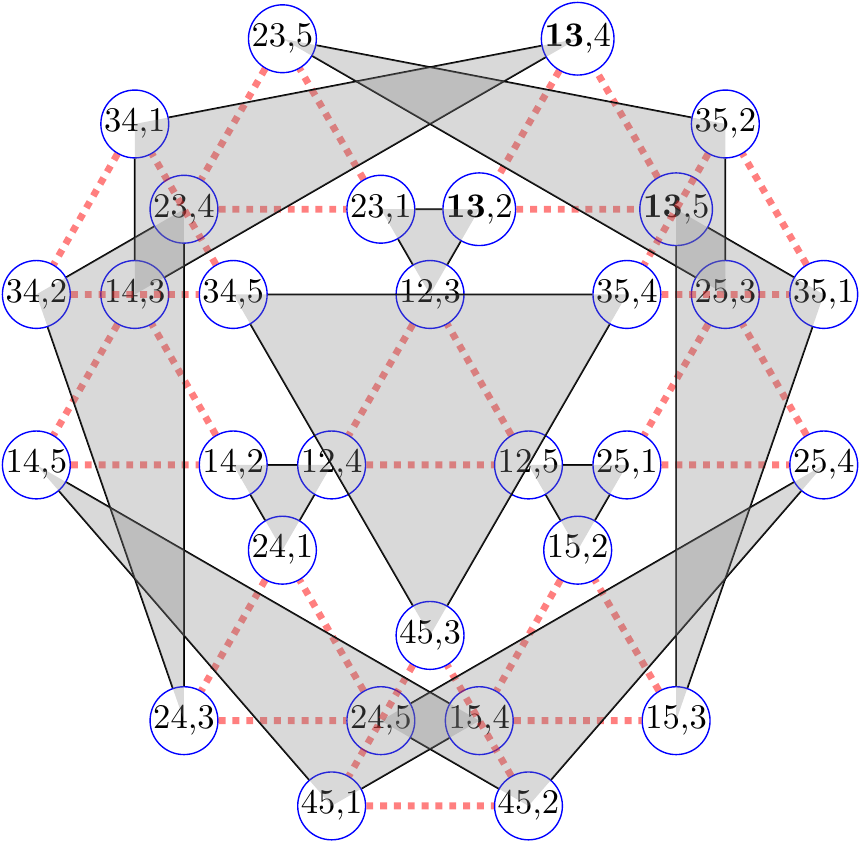}
    \caption{The quasi-Johnson graph $G$ when $N=5$ and $|S|=r=2$. Note that $\binom{N}{r} (N-r) = \binom{N}{r+1} (r+1) = 10\times 3$, so that the two types of cliques $A_S$ and $B_{[S,y]}$ are of the same size $3$ and of the same quantity $10$. The vertex $(S,y)=(\{1,2\},3)$ in $G$, for example, is represented by the blue circle with shorthand ``$12,3$'' inside. The cliques $A_S$ are represented by equilateral triangles with red dotted edges. The cliques $B_{[S,y]}$ are represented by shaded gray triangles with black solid edges, $4$ of them are equilateral and the other $6$ triangles are obtuse. Suppose $(1,3)$ is the indices of the colliding pair in $x=(x_1,\ldots,x_5)$. Then the three nodes in the upper right corner with bold numbers ``$\textbf{13}$'' inside are the marked vertices.}
    \label{fig:johnson}
\end{figure}

If we denote the equal-superposition of all the marked vertices by $\ket{T}$, and the equal-superposition of all vertices by $\ket{\psi_0} := \frac{1}{\sqrt{|V|}} \sum_{(S,y)\in V} \ket{S,y}$, then the goal of quantum walk search on quasi-Johnson graph $G$ is to transform $\ket{\psi_0}$ to $\ket{T}$ by applying an appropriate sequence consists of $U_A(\theta_1),U_B(\theta_2)$ and $R_T(\alpha)$.

\begin{algorithm}[htb]
    \SetKwFunction{}{}
    \SetKwInOut{KWProcedure}{Procedure}
    \SetKwInput{Runtime}{Runtime}
    \caption{Exact quantum algorithm for Problem~\ref{problem:at_most_1}}
    \label{alg:main_block}
    \LinesNumbered

    \KwIn {index oracle $O_x$ given in Problem~\ref{problem:at_most_1} [effect shown in Eq.~\eqref{eq:oracle}]}
    \KwOut {indices $(i_1,i_2)$ of the colliding pair; or ``all distinct”.}
    \KWProcedure{}
    
    \begin{enumerate}[1.]
    \item Prepare the equal-superposition  of all vertices $\ket{\psi_0} = \frac{1}{\sqrt{|V|}} \sum_{(S,y)\in V} \ket{S,y}$ in the first register.
    
    \item Call decoding procedure $\mathcal{P}$ on $\ket{\psi_0}$ to obtain the ascending index sequence $(i_1,\ldots,i_r)$ of $S$, and query $O_x$ for $r$ times to obtain $\frac{1}{\sqrt{|V|}} \sum_{(S,y)\in V} \ket{S,y}\ket{x_{i_1},\ldots,x_{i_r},0} $. Then apply $\mathcal{P}$ again to erase the ancillary register storing the index sequence.

    \item Iterate the following steps for $t_1$ times:
    \begin{enumerate}[(a)]
        \item Apply the marking operator $R_T(\alpha_1)$ which multiplies $\ket{S,y}\ket{x'_1,\ldots,x'_{r},0}$ by $e^{i\alpha_1}$ if the second register contains the colliding pair.
        \item Iterate the following steps for $ct_2$ times:
        \begin{enumerate}[(i)]
            \item Apply $U_{A}(\theta_1)$ [see Eq.~\eqref{eq:U_A_def}] to the first register.
            \item Call $\mathcal{P}$ on $\ket{S,y}$ to obtain $\bar{y}$, and query $O_x$ with $\bar{y}$, storing the answer in the last slot of the second register $\ket{x'_{r+1}}$. Then apply $\mathcal{P}$ again to erase the ancillary register storing the index sequence.
            \item Apply $U_B^{\text{EXT}}(\theta_2)$ [see Eq.~\eqref{eq:U_B_def}] on the two working registers.
            \item Repeat step (ii) to set $\ket{x'_{r+1}}$ back to $\ket{0}$.
        \end{enumerate}
        \item Apply the marking operator $R_T(\alpha_2)$ with a different angle $\alpha_2$.
        \item Repeat step (b).
    \end{enumerate}
    \item Measure the two working registers in computational basis, and based on the outcome $S$ and $(x_{i_1'},\ldots,x_{i_r'})$, obtain indices $(i_1,i_2)$ of the colliding pair; or output ``all distinct” if $(x_{i_1'},\ldots,x_{i_r'})$ are all distinct.
    \end{enumerate}
    
    p.s. The $6$ parameters $t_1,\alpha_1,\alpha_2,ct_2,\theta_1,\theta_2$ is set according to Eqs.~\eqref{eq:inner_iter1}-\eqref{eq:cos_phi_def}.
    
\end{algorithm}

\subsection{Parameters in the algorithm}
Algorithm~\ref{alg:main_block} contains $6$ parameters $\alpha_1,\alpha_2,\theta_1,\theta_2,t_1,c\cdot t_2$ to be set below. We justify their values in the next section.
The iteration  number of the inner loop [step 3.(b)] is
\begin{equation}\label{eq:inner_iter1}
    c\cdot t_2 = 10\cdot \lceil \frac{\pi}{2}\sqrt{r} \rceil.
\end{equation}
The rotation parameters $\theta_1,\theta_2$ in the inner loop are solutions to the following equations:
\begin{equation}\label{eq:inner_angle1}
    \cos\phi_1 = -\cos(1-\frac{2}{c})\frac{\pi}{t_2}, \quad \cos\phi_2 = -\cos\frac\pi{t_2},
\end{equation}
where $\cos\phi_i = \cos\frac{\theta_1+\theta_2}{2} + 2\sin\frac{\theta_1}{2} \sin\frac{\theta_2}{2}\, \lambda_i$, and
$\lambda_i=\frac{i(N+1-i)}{(N-r)(r+1)}$.
Let $\beta :=-( c\cdot t_2 \frac{\theta_1+\theta_2}{2} \mod 2\pi)$, $\lambda=\frac{r(r-1)}{N(N-1)}$, and $\phi_0 := |4\arcsin(\sqrt{\lambda}\sin\frac{\beta}{2}) \mod [-\frac{\pi}{2},\frac{\pi}{2}]|$, where the notation ``$x \mod [-\frac{\pi}{2},\frac{\pi}{2}]$'' means  adding $x$ with an appropriate integer multiples $l$ of $\pi$, such that $x+l\pi\in [-\frac{\pi}{2},\frac{\pi}{2}]$. Then the iteration number of the outer loop [step 3] is
\begin{equation}\label{eq:outer_iter1}
    t_1 = \lceil \frac{\pi}{\phi_0} \rceil.
\end{equation}
The angle parameters $\alpha_1,\alpha_2$ in the outer loop are solutions to the following equations:
\begin{align}
    0 &= 1- 2\lambda \frac{\tan(t_1\phi)}{\sin\phi} ( -\sin\frac{\alpha_1}{2} \cos\frac{\alpha_2}{2} \sin\beta + 2\sin\frac{\alpha_1}{2} \sin\frac{\alpha_2}{2} \sin^2\frac{\beta}{2} (1-2\lambda) ), \label{eq:real1}\\
    0 &=-(1-2\lambda)\sin\beta \tan\frac{\alpha_1}{2} \tan\frac{\alpha_2}{2} + \cos\beta \tan\frac{\alpha_1}{2} +(2\lambda+(1-2\lambda)\cos\beta) \tan\frac{\alpha_2}{2}  +\sin\beta \label{eq:im1},
\end{align}
where $\phi$ satisfies
\begin{align}\label{eq:cos_phi_def}
    \cos{\phi} &= \cos(\frac{\alpha_1+\alpha_2}{2}+\beta) + 2\lambda (\sin\frac{\alpha_1+\alpha_2}{2} \sin\beta -4\sin\frac{\alpha_1}{2} \sin\frac{\alpha_2}{2} \sin^2\frac{\beta}{2}) 
    \notag \\ &\qquad
    + 8\lambda^2 \sin\frac{\alpha_1}{2} \sin\frac{\alpha_2}{2} \sin^2\frac{\beta}{2}.
\end{align}

\section{Analysis of the algorithm}
We first analyze the query complexity of Algorithm~\ref{alg:main_block}. From Eq.~\eqref{eq:outer_iter1}, the iteration number $t_1$ of the outer loop is related to angle $\beta/2$. Note that $\frac{\beta}{2}$ tends to a constant which is irrelevant to $N$ when $N\to\infty$ from Eq.~\eqref{eq:infty_beta}. Therefore, the inner loop is iterated for $2t_1\in O(N/r)$ times. As the inner loop queries $O_x$ for $2ct_2\in O(\sqrt{r})$ times, and the preparation of $\ket{\psi_0}$ needs $r$ queries, the total query complexity is $O(N/\sqrt{r}+r)=O(N^{2/3})$ by Eq.~\eqref{eq:r_def}.

We now analyze the correctness of Algorithm~\ref{alg:main_block}. Note that $\ket{\psi_0}$ remains unchanged when the items of $x$ are all distinct, because $R_T(\alpha)$ now becomes the identity operator $I$, and $\ket{\psi_0}$ is invariant under $U_A(\theta_1)$ and $U_B(\theta_2)$. Hence, Algorithm~\ref{alg:main_block} is correct for this case. We thus assume that $x$ contains a single colliding pair in the following three subsections.

\subsection{Reducing to $5$-dimensional subspace}\label{sec:reduction}
The reduced $5$-dimensional subspace $\mathcal{H}_0$ is spanned by five equal-superpositions of the vertices in five different groups: $\eta_0^0,\eta_0^1,\eta_1^0,\eta_1^1,\eta_2^0$, where $\eta_l^j = \{(S,y)\in V: |S\cap K|=l,\, |\{y\}\cap K|=j \}$. That is, we group the vertices of $G$ into five different groups by the $5$ different situations that $(S,y)$ intersects with indices $\{i_1,i_2\}=:K$ of the colliding pair. The size of each groups, by basic combinatorics, is
\begin{equation}\label{eq:set_size}
    |\eta_l^0|=\binom{k}{l}\binom{N-k}{r-l}(N-r-(k-l)), \quad |\eta_l^1|=\binom{k}{l}\binom{N-k}{r-l}(k-l),
\end{equation}
where $k=2$ is the size of $K$. We denote by $\ket{T}$ the target state $\ket{\eta_2^0}$, as all vertices in this group contains indices $\{i_1,i_2\}$ of the colliding pair. Clearly $\ket{\psi_0}\in\mathcal{H}_0$, and the square root of the proportion of targets to all vertices is
\begin{equation}\label{eq:psi_0}
    \braket{T|\psi_0} = \sqrt{\frac{r(r-1)}{N(N-1)}}.
\end{equation}

We now show that the effect of step 3.(b) of Algorithm~\ref{alg:main_block} in $\mathcal{H}_0$ is
\begin{equation}\label{eq:u}
    u = U_B(\theta_2) U_A(\theta_1) = (I-(1-e^{i\theta_2})BB^\dagger) \cdot ( I-(1-e^{i\theta_1})AA^\dagger ),
\end{equation}
where
\begin{equation}
    A\circ A=
\begin{bmatrix}
{1-\frac{2}{N-r}} & 0 & 0 \\
{\frac{2}{N-r}} & 0 & 0 \\
0 & {1-\frac{1}{N-r}} & 0 \\
0 & {\frac{1}{N-r}} & 0 \\
0 & 0 & 1
\end{bmatrix},
\quad B\circ B=
\begin{bmatrix}
1 & 0 & 0 \\
0 & \frac{1}{r+1} & 0 \\
0 & 1-\frac{1}{r+1} & 0 \\
0 & 0 & \frac{2}{r+1} \\
0 & 0 & 1-\frac{2}{r+1}
\end{bmatrix}.
\end{equation}
The  notation $\circ$ here is the Hadamard (element-wise) product. We use it because the square root notation $\sqrt{*}$ takes up too much space.

\textbf{Proof (Calculation):} Consider $U_A(\theta_1) = I-(1-e^{i\theta_1}) \sum_{S\in S_r} \ket{A_S}\bra{A_S}$ first. 
Recall from Eq.~\eqref{eq:As} that $\ket{A_S} = \frac{1}{\sqrt{N-r}} \sum_{y\in[N]\setminus y} \ket{S,y}$ is determined by $S$. Thus $\braket{A_S | \eta_l^j} \neq 0$ only if $|S\cap K|=l$, where $l\in \{0,1,2\}$. 
To determine its value, note that $\ket{S,y}$ in $\ket{\eta_l^j}$ has to have the same $S$ as that of $A_S$ to produce nonzero inner-product. Therefore, we have 
\begin{align}
\braket{A_S | \eta_l^j} &= \frac{\delta_{|S\cap K|,l}}{\sqrt{N-r}\sqrt{|\eta_l^j|}} 
\times\begin{cases}
N-r-(k-l), & j=0 \\
k-l, & j=1
\end{cases}\\
&= \frac{\delta_{|S\cap K|,l}}{\sqrt{N-r} \sqrt{\binom{k}{l} \binom{N-k}{r-l}}}
\times \begin{cases}
\sqrt{N-r-(k-l)}, & j=0 \\
\sqrt{k-l}, & j=1
\end{cases}. \label{eq:sumSr}
\end{align}
We now consider the sum $\sum_{S\in S_r} \braket{\eta_{l'}^{j'} | A_S} \braket{A_S | \eta_l^j}$. Note that it's nonzero only if $l=l'$, and thus the sum now reduces to be over those $S\in S_r$ that satisfy $|S\cap K|=l$. There are $\binom{k}{l} \binom{N-k}{r-l}$ of this kind, cancelling the denominator in Eq.~\eqref{eq:sumSr}. Hence, the effect of $U_A(\theta_1)$ in its invariant subspace $\{\ket{\eta_l^0}, \ket{\eta_l^1}\}$ is
\begin{equation}
    I-(1-e^{i\theta_1})
\begin{bmatrix}
\sqrt{1-\frac{k-l}{N-r}}\\
\sqrt{\frac{k-l}{N-r}}
\end{bmatrix}
\begin{bmatrix}
\sqrt{1-\frac{k-l}{N-r}}, &
\sqrt{\frac{k-l}{N-r}}
\end{bmatrix},
\end{equation}
which degenerates to the scalar multiplication $e^{i\theta_1} \ket{\eta_k^0}\bra{\eta_k^0}$ when $l=k$. The effect of $U_A(\theta_1)$ in $\mathcal{H}_0$ as shown in Eq.~\eqref{eq:u} now follows easily. The effect of $U_B(\theta_2)$ can be determined in a similar fashion. \hfill $\blacksquare$

The effect of $R_T(\alpha)$ in $\mathcal{H}_0=\text{span} \{ \ket{\eta_0^0} , \ket{\eta_0^1} , \ket{\eta_1^0}, \ket{\eta_1^1}, \ket{\eta_2^0}  \}$ that appeared in step 3.(a) and 3.(c) of Algorithm~\ref{alg:main_block} is obviously
\begin{equation}\label{eq:Rt}
    R_T(\alpha)=\text{diag}\left( [1,1,1,1,e^{i\alpha}] \right).
\end{equation}
Now the goal of Algorithm~\ref{alg:main_block}, in a nutshell, is to achieve the following equation
\begin{equation}
    1 = \bra{T} \big(u^{c\,t_2} R_T(\alpha_2) \cdot u^{c\,t_2} R_T(\alpha_1) \big)^{t_1} \ket{\psi_0},
\end{equation}
by appropriately setting the parameters therein.

\subsection{Constructing phase rotation about the initial state}\label{sec:clock}

In this subsection, we show how to construct the following phase rotation $R_{\psi_0}(\beta)$ about the initial state $\ket{\psi_0}$, by iterating the quantum walk operator $u= U_B(\theta_2) U_A(\theta_1)$ for $c\cdot t_2$ times.
\begin{equation}
    R_{\psi_0}(\beta) := I-(1-e^{i\beta}) \ket{\psi_0}\bra{\psi_0}
\end{equation}
We first justify the values of $\theta_1,\theta_2$ and $c\cdot t_2$ as shown respectively in Eq.~\eqref{eq:inner_angle1} and Eq.~\eqref{eq:inner_iter1} by the following lemma.

\begin{lemma}\label{lem:inner_effect}
By setting the parameters $\theta_1,\theta_2$ in $u= U_B(\theta_2) U_A(\theta_1)$ and $c\cdot t_2$ according to Eq.~\eqref{eq:inner_angle1} and Eq.~\eqref{eq:inner_iter1} respectively, the effect of $u^{ct_2}$ in $\mathcal{H}_0$ is the following:
\begin{equation}\label{eq:psi_rotate}
    I-(1- e^{i\,ct_2\frac{\theta_1+\theta_2}{2}} )\ket{\psi_0}\bra{\psi_0}.
\end{equation}
\end{lemma}

\textbf{Attention:} Eq.~\eqref{eq:psi_rotate} is defined in the subspace $\mathcal{H}_0$, but not in the whole space spanned by all the vertices of $G$. In fact, to construct Eq.~\eqref{eq:psi_rotate} in the whole space, we claim that one must iterate $u$ for at least $r$ times rather than $ct_2\in O(\sqrt{r})$ times:
if $S_1\cup \{y_1\}$ and $S_2\cup \{y_2\}$ are disjoint (this is possible, since $|S|=r=N^{2/3}\ll N$), then one needs to iterate $u$ for at least $r$ times in order to `transform' $\ket{\psi}$ to $\ket{\varphi}$. `Transform' here means that after $\ket{\psi}$ is changed to some superposition, $\ket{\varphi}$ appears in this superposition. 
This is because $U_A(\theta_1)$ and $U_B(\theta_2)$ are diffusion operators within clique $A_S$ and $B_{[S,y]} $ respectively. Thus a step of quantum walk $u=U_B(\theta_2) U_A(\theta_1)$ will only update one index in $S$ for a fixed $(S,y)$.
Eq.~\eqref{eq:psi_rotate} defined in the whole space can certainly `transform' $\ket{\psi}$ to $\ket{\varphi}$ in one step, and thus the claim follows.

\begin{proof}
Similar to the spectral lemma in \cite{Szegedy_03}, we first perform singular value decomposition on the discriminant matrix $D:=A^\dagger B$ of order $3$ as $D=WSV^\dagger = \sum_{i=0}^2 \lambda_i \ket{w_i}\bra{v_i}$. The square $\lambda_i^2$ of the $3$ singular values are
\begin{equation}
     \left[1,\ \left(1-\frac{1}{N-r}\right) \left(1-\frac{1}{r+1}\right),\ \left(1-\frac{2}{N-r}\right) \left(1-\frac{2}{r+1}\right) \right]
     =:
     [\cos^2\frac{\gamma_0}{2},\cos^2\frac{\gamma_1}{2},\cos^2\frac{\gamma_2}{2}].
     \footnote{The singular values shows a particular pattern, i.e. $\cos^2\frac{\gamma_i}{2} = (1-i/(N-r))\cdot(1-i/(r+1)) $, for all the $D=A^\dagger B$ with order $k\geq 1$. We are not able to prove this, and have posted it as an open question in \cite{web_question}.}
\end{equation}
Then we have the following two key equations, which lies at the heart of Jordan's lemma \cite{Jordan_1875}:
\begin{align}
& \Pi_A \cdot B\ket{v_i} := AA^\dagger B\ket{v_i} = AD\ket{v_i}= \cos\frac{\gamma_i}{2} \cdot A \ket{w_i}, \\
& \Pi_B \cdot A\ket{w_i} := BB^\dagger A\ket{w_i} = B(\bra{w_i}D)^\dagger = \cos\frac{\gamma_i}{2} \cdot B \ket{v_i}.
\end{align}
It can be verified that $A\ket{w_0} = B\ket{v_0} = \ket{\psi_0}$. Thus we obtain the three {\it common} invariant subspaces $\{ \ket{\psi_0} \}$, $\text{span}\{A\ket{w_1}, B\ket{v_1} \}$ and $\text{span}\{A\ket{w_2}, B\ket{v_2} \}$ of
\begin{equation}\label{eq:u}
    u=(I-(1-e^{i\theta_2})\Pi_B) \cdot ( I-(1-e^{i\theta_1})\Pi_A ).
\end{equation}

The effect of $u$ in the subspace $\text{span}\{A\ket{w_i}, B\ket{v_i} \}$ ($2$-dimensional space) is the composition of two 3D-rotations: a rotation about $A\ket{w}$ by angle $-\theta_1$, followed by a rotation about $B\ket{w}$ by $-\theta_2$, and the angle between the two rotation axes is $\gamma_i$. We articulate this in the following.
As $\bra{w_i}A^\dagger B\ket{v_i} =\cos\gamma_i/2$, we might as well let $B\ket{v_i} = \cos\gamma_i/2 \ A\ket{w_i} + \sin\gamma_i/2\ A\ket{w_i}^\perp$, where $A\ket{w_i}^\perp$ is perpendicular to $A\ket{w_i}$ and lies in the same qubit space. Simple calculation shows that the effect of $U_A(\theta_1)$ on the orthonormal basis $\{A\ket{w_i},A\ket{w_i}^\perp\}$ is
\begin{equation}
    U_A(\theta_1) = 
\begin{bmatrix}
e^{i\theta_1} & 0 \\
0 & 1
\end{bmatrix}
= e^{i\theta_1/2} R_{Aw_i}(-\theta_1/2).
\end{equation}
The notation
\begin{equation}\label{eq:rotate_def}
    R_{\vec{n}}(\phi) := \cos\phi I -i\sin\phi\, \vec{n} \cdot \vec{\sigma}
\end{equation}
represents a rotation on the Bloch sphere about axis $\vec{n}$ by angle $2\phi$ obeying the right-handed rule, and $\vec{n} \cdot \vec{\sigma} := n_x X+n_y Y+n_z Z$, where $X,Y,Z$ are the three Pauli matrices. $U_A(\theta_1)$ thus represents a rotation about $Aw_i=[0,0,1]$ by angle $-\theta_1$. Similarly, the effect of $U_B(\theta_2)$ on the orthonormal basis $\{A\ket{w_i},A\ket{w_i}^\perp\}$ is
\begin{equation}
    U_B(\theta_2)=
\begin{bmatrix}
1-(1-e^{i\theta_2}) \cos^2 \frac\gamma 2 & -(1-e^{i\theta_2}) \sin\frac\gamma 2 \cos\frac\gamma 2 \\
-(1-e^{i\theta_2}) \sin\frac\gamma 2 \cos\frac\gamma 2 & 1-(1-e^{i\theta_2}) \sin^2 \frac\gamma 2
\end{bmatrix}
= e^{i\theta_2/2} R_{Bv_i}(-\theta_2/2),
\end{equation}
which represents a rotation about $Bv_i=[\sin\gamma_i,0,\cos\gamma_i]$ by angle $-\theta_2$. Thus the angle between the two rotation axes $Aw_i$ and $Bv_i$ is $\gamma_i$.

By the composition lemma of two rotations (Exercise~4.15 in \cite{nielsen_chuang_2010}), the effect of $u$ (shown in Eq.~\eqref{eq:u}) in the subspace $\text{span}\{A\ket{w_i}, B\ket{v_i} \}$ as the composition of two rotations is another rotation $e^{i\frac{\theta_1+\theta_2}{2}} R_{\vec{n}}(\phi_i)$, whose rotation angle $\phi_i$ satisfies:
\begin{equation}\label{eq:phi_i}
    \cos\phi_i = \cos\frac{\theta_1+\theta_2}{2} + 2\sin\frac{\theta_1}{2} \sin\frac{\theta_2}{2}\, \lambda_i,
\end{equation}
where $\lambda_i:=\sin^2\frac{\gamma_i}{2}=\frac{i(N+1-i)}{(N-r)(r+1)}$. Note that Eq.~\eqref{eq:inner_angle1} is equivalent to
\begin{equation}\label{eq:inner_angle2}
    \phi_1 = \pi \pm (1-\frac{2}{c})\frac{\pi}{t_2}, \quad \phi_2 = \pi \pm \frac{\pi}{t_2}.
\end{equation}
Therefore $c t_2 \phi_i \equiv 0 \mod 2\pi$ ( $c=10$ is even). From the composition lemma of two rotations and mathematical induction, we have
\begin{equation}\label{eq:rotate_composition}
    R_{\vec{n}}^k(\phi) = R_{\vec{n}}(k\phi).
\end{equation}
Note that $R_{\vec{n}}(\phi)=I$ when $\phi \equiv 0 \mod 2\pi$. Hence, the effect of $u^{ct_2}$ on the three invariant subspace $\{ \ket{\psi_0} \}$, $\text{span}\{A\ket{w_1},A\ket{w_1}^\perp\}$ and $\text{span}\{A\ket{w_2},A\ket{w_2}^\perp\}$, up to the irrelevant global phase $e^{i\, ct_2\frac{\theta_1+\theta_2}{2}}$, are $e^{i\, ct_2\frac{\theta_1+\theta_2}{2}}I,\ I$ and $I$, respectively. This completes the proof of Lemma~\ref{lem:inner_effect}.
\end{proof}

The intuition behind Eq.~\eqref{eq:inner_angle2} can be interpreted as the minute hand ($\phi_2$) catching up with the hour hand ($\phi_1$) on a watch [considering the case that `$\pm$' is `$+$' in Eq.~\eqref{eq:inner_angle2}].
As $\cos\phi_i$ is linear in $\lambda_i$, the rotation angle $\phi_i$ will never be equal since $\lambda_1$ and $\lambda_2$ differs. We let the difference between $\phi_1$ and $\phi_2$ to be $\frac{2\pi}{ct_2}$, and let the difference between $\phi_2$ and $\pi$ to be $\frac{\pi}{t_2}$,
so that after $ct_2$ iterations, $\phi_1$ is slower than $\phi_2$ for an entire cycle, and $\phi_2$ is faster than $\pi$ for $c/2$ cycle. Hence, $c t_2 \phi_i \equiv 0 \mod 2\pi$ as desired.
More importantly, setting iteration number $ct_2$ according to Eq.~\eqref{eq:inner_iter1} guarantees that the equations \eqref{eq:inner_angle2} about $\theta_1,\theta_2$ have a solution. This is shown in the following lemma.

\begin{lemma}\label{thm:inner_exist}
When $N\geq 5, c=10$, and $t_2 \geq \frac{\pi}{2} \sqrt{r}$,
there exist solutions $\theta_1,\theta_2$ to Eq. \eqref{eq:inner_angle1}. One of the solutions satisfies $\frac{\theta_1+\theta_2}{2} =\pi - d \frac\pi{t_2}$, where the intermediate variable $d$ has two properties: $cd\pi$ is not an integer multiple of $2\pi$; and when $N\to \infty$, $d \to \sqrt{c^2-8(c-1)}/c =\sqrt{7}/5$.
This solution can be obtained by solving Eq.~\eqref{eq:first} first for $d$, and then substitute it into Eq.~\eqref{eq:second} to obtain $\theta_1$, and finally let $\theta_2=2(\pi-d\pi/t_2)-\theta_1$.
\end{lemma}

\begin{proof}
See Appendix \ref{proof:inner_exist}.
\end{proof}

According to Lemma~\ref{thm:inner_exist}, when $N\to \infty$, the effect of $u^{ct_2}$ multiplies $\ket{\psi_0}$ by a relative phase
\begin{equation}\label{eq:infty_beta}
    e^{-cdi\pi}=e^{ -2\sqrt{7}i\pi} \approx e^{ -1.29 i\pi},
\end{equation}
and the claim that $cd\pi$ is not an integer multiple of $2\pi$ guarantees $u^{ct_2}\neq I$.

\subsection{Exact quantum search with arbitrarily given phase rotation}\label{sec:FXR}

In this subsection, we will transform the initial state $\ket{\psi_0}$ to $\ket{T}$ in $\mathcal{H}_0$ by appropriately setting the angle $\alpha$ in $R_T(\alpha)$, with the help of the FXR method \cite{FXR}.

Let $\beta =-( c\cdot t_2 \frac{\theta_1+\theta_2}{2} \mod 2\pi)=(cd\pi \mod 2\pi)$, then by Eq.~\eqref{eq:Rt} and Lemma~\ref{lem:inner_effect}, the effect of $R_T(\alpha)$ and $u^{ct_2}$ in $\mathcal{H}_0$ are
\begin{equation}
    I-(1-e^{i\alpha})\ket{T}\bra{T} =: S_o(\alpha),
    \quad
    I-(1-e^{-i\beta})\ket{\psi_0}\bra{\psi_0} =: S_r(\beta),
\end{equation}
respectively. Recall from Eq.~\eqref{eq:psi_0} that $\braket{T|\psi_0} = \sqrt{\frac{r(r-1)}{N(N-1)}} =: \sqrt{\lambda}$. 
Thus we might as well let $\ket{\psi_0} = \sqrt{1-\lambda} \ket{R} + \sqrt{\lambda}\ket{T}$, where $\ket{R}\in \text{span}\{\ket{T},\ket{\psi_0}\}$ and is perpendicular to $\ket{T}$.
We denote by $\mathcal{H}_1 := \text{span}\{\ket{R},\ket{T}\}$ the $2$-dimensional invariant subspace of $u^{ct_2}R_T(\alpha) = S_r(\beta) S_o(\alpha)$. Note that $S_r(\beta) S_o(\alpha)$ is the generalized Grover's operator $G(\alpha,\beta)$, where $S_o$ is the `oracle' operator, and $S_r$ is the diffusion operator.

Just as the calculation step of Theorem 2 in \cite{FXR}, by repeatedly using the composition lemma of two rotations, we find that the effect of 
\begin{equation}
    F_{xr,\beta}:= G(\alpha_2,\beta) G(\alpha_1,\beta),
\end{equation}
which is now the effect of one iteration of the outer loop (step 3 in Algorithm~\ref{alg:main_block}),
is a rotation $e^{i(\frac{\alpha_1+\alpha_2}{2}-\beta)} R_{\vec{n}}(\phi)$ on the Bloch sphere of $\mathcal{H}_1$. The rotation angle $\phi$ satisfies Eq.~\eqref{eq:cos_phi_def}, and the axis $\vec{n}$ are as follows:
\begin{align}
    \sin\phi\, n_x &= 2\sqrt{\lambda(1-\lambda)} (\cos\frac{\alpha_1}{2} \cos\frac{\alpha_2}{2} \sin\beta - 2(1-2\lambda) \cos\frac{\alpha_1}{2} \sin\frac{\alpha_2}{2} \sin^2\frac{\beta}{2}), \label{eq:sinnx1}\\
    \sin\phi\, n_y &= 2\sqrt{\lambda(1-\lambda)} ( -\sin\frac{\alpha_1}{2} \cos\frac{\alpha_2}{2} \sin\beta + 2(1-2\lambda) \sin\frac{\alpha_1}{2} \sin\frac{\alpha_2}{2} \sin^2\frac{\beta}{2} ), \label{eq:sinny1}\\
    \sin\phi\, n_z &= \sin\frac{\alpha_1}{2} \cos\frac{\alpha_2}{2} \cos\beta +\cos\frac{\alpha_1}{2} \sin\frac{\alpha_2}{2}
    \notag \\ &\qquad
    +(1-2\lambda) \cos\frac{\alpha_1+\alpha_2}{2} \sin\beta  -2(1-2\lambda)^2 \cos\frac{\alpha_1}{2} \sin\frac{\alpha_2}{2} \sin^2\frac{\beta}{2} \label{eq:sinnz1}.
\end{align}

Transforming $\ket{\psi_0}$ to $\ket{T}$ in $\mathcal{H}_1$ is then equivalent to the equation $0 = \bra{R} R^k_{\vec{n}}(\phi) \ket{\psi_0}$. By setting the equation's real and imaginary part both to be zero, and with the help of Eq.~\eqref{eq:rotate_composition}, we obtain the equations that parameters $(k=t_1,\alpha_1,\alpha_2)$ need to satisfy:
\begin{align}
0 &= \sqrt{1-\lambda} \cos{k\phi} - \sqrt{\lambda} \sin{k\phi}\ n_y, \label{real1} \tag{real 1}\\
0 &= \sqrt{1-\lambda}\ n_z + \sqrt{\lambda}\ n_x. \label{im1} \tag{im 1}
\end{align}
Substituting Eqs.~\eqref{eq:sinnx1}-\eqref{eq:sinnz1} into the above two equations then leads to Eqs.~\eqref{eq:real1},\eqref{eq:im1}.

Let $\phi_0 = |4\arcsin(\sqrt{\lambda}\sin\frac{\beta}{2}) \mod [-\frac{\pi}{2},\frac{\pi}{2}]|$, then Theorem 2 in \cite{FXR} shows that the above equations possess a solution when $k>\pi/\phi_0$. This explains the value of $t_1$ shown in Eq.~\eqref{eq:outer_iter1}.

\subsection{Query lower bound}\label{sec:lower_bound}
\begin{lemma}\label{thm:lower}
The deterministic quantum query complexity of Problem~\ref{problem:at_most_1} is $\Omega(N^{2/3})$.
\end{lemma}
\begin{proof}
The easier decision problem of Problem~\ref{problem:at_most_1}, i.e. to determine whether elements in $(x_1,x_2,\ldots,x_N)\in [M]^N$ are all distinct, is in fact the evaluation of partial Boolean function $f:\mathcal{D}\to \{0,1\}$, whose domain $\mathcal{D}\subseteq [M]^N$ consists of two parts: $f^{-1}(1)=\{x\in[M]^N : x\ \text{contains a single colliding pair} \}$, and $f^{-1}(0)=\{x\in[M]^N : \text{elements in}\ x\ \text{are all distinct} \}$. 
We therefore need only to show that the deterministic quantum query complexity of evaluating $f$ is $\Omega(N^{2/3})$. 
Let $Q_{\epsilon}(f)$ be the bounded-error quantum query complexity of evaluating $f$. Bounded-error means to output the correct answer with probability $1-\epsilon$, where $0\leq\epsilon<1/2$. According to Theorem 2 in \cite{hoyer2007negative},
\begin{equation}
    Q_{\epsilon}(f)\geq \frac{1-2\sqrt{\epsilon(1-\epsilon)}}{2} \text{ADV}^\pm (f),
\end{equation}
where $\text{ADV}^\pm (f)$ is the negative-weight adversary lower bound of $f$. Hence, by letting $\epsilon=0$, we know that the deterministic quantum query complexity of evaluating $f$ is at least $\text{ADV}^\pm (f)/2$. The definition of $\text{ADV}^\pm (f)$ is
\begin{equation}
    \text{ADV}^\pm (f) = \max_{\Gamma\in \mathcal{M}} \frac{||\Gamma||}{\max_{i\in[n]} ||\Gamma \circ \Delta_i||},
\end{equation}
where $\mathcal{M}$ are the set of all the nonzero real matrices, whose row indices composed of elements in $f^{-1}(1)$, and whose column indices composed of elements in  $f^{-1}(0)$; element $(x,y)$ in the $0-1$ matrix $\Delta_i\in\mathcal{M}$ equals $1$ if and only if $x_i\neq y_i$; notation $\circ$ is the Hadamard product, and $||\cdot||$ is the spectral norm.

It has been shown respectively in \cite{belovs2012adversary} and \cite{rosmanis2014quantum} that when $M=\Omega(N^2)$ and $M=N$, the lower bound of $\text{ADV}^\pm (f)$ are both $\Omega(N^{2/3})$.
The proof in \cite{rosmanis2014quantum} does not follow the idea in \cite{belovs2012adversary} of emending the adversary matrix $\Gamma$ into some bigger matrix related to the Hamming association scheme, and therefore it breaks the restriction of $M=\Omega(N^2)$. It's also said in \cite{rosmanis2014quantum} that its proof can be generalized to any $M\geq N$, thus the deterministic quantum query complexity of Problem~\ref{problem:at_most_1} is $\Omega(N^{2/3})$.
\end{proof}

\section{Conclusion and future work}
In this paper, we have designed an exact quantum algorithm with query complexity $\Theta(N^{2/3})$ that solves the promised element distinctness  problem with at most one colliding pair. The algorithm utilize Jordan's lemma about common invariant subspaces, and the recently proposed FXR method for exact quantum search with arbitrary phase rotations, so as to achieve certainty of success. In contrast, previous algorithms based on the analysis of asymptotic behaviours of eigenvalues of quantum walk search  could not achieve 100\% success.
Whether the idea in this paper could be applied to other  quantum walk search problem is a possible direction for future research.
\bibliographystyle{unsrt}
\bibliography{reference}


\begin{appendices}

\section{Proof of Lemma \ref{thm:inner_exist}}\label{proof:inner_exist}

We first rearrange Eq.~\eqref{eq:inner_angle1} into a form that is easier to analyze the existence of solutions.
Subtracting the first equation of Eq.~\eqref{eq:inner_angle1} by the second equation, and substituting the expression of $\cos\phi_i$ into them, we obtain
\begin{equation}\label{eq:inner_angle3}
    2\sin\frac{\theta_1}{2} \sin\frac{\theta_2}{2} \Delta = \delta_c, \quad 2\sin\frac{\theta_1}{2} \sin\frac{\theta_2}{2}\, \lambda_2 = \cos\left( \pi-\frac{\theta_1+\theta_2}{2} \right) - \cos\frac\pi{t_2},
\end{equation}
where $\Delta = \lambda_2-\lambda_1 = \frac{N-2}{(N-r)(r+1)}$, and $\delta_c = \cos((1-2/c)\pi/t_2) - \cos(\pi/t_2)$.
We then introduce the intermediate variable $d$ satisfying:
\begin{equation}\label{eq:relation_d_theta}
    \pi-\frac{\theta_1+\theta_2}{2} = d \frac\pi{t_2}.
\end{equation}
Substituting (Eq.1) of Eq.~\eqref{eq:inner_angle3} by the division of (Eq.2)/(Eq.1), and dividing both sides of (Eq.2) by $\lambda_2$, we have the following equations about $d$ and $\theta_1$:
\begin{align}
    \frac{\cos(d\frac\pi{t_2})-\cos(\frac\pi{t_2})}{\cos((1-\frac{2}{c})\frac\pi{t_2}) - \cos(\frac\pi{t_2})} &= 2(1+\frac{1}{N-2}), \label{eq:first} \\
    \cos(d\frac\pi{t_2}) - \cos(d\frac\pi{t_2}+\theta_1) &= \frac{1}{\lambda_2} \left(\cos(d\frac\pi{t_2})-\cos(\frac\pi{t_2})\right) \label{eq:second},
\end{align}
where we have used the identity $2\sin\frac{\theta_1}{2} \sin\frac{\theta_2}{2} = \cos(d\frac\pi{t_2}) - \cos(d\frac\pi{t_2}+\theta_1)$.
As stated in Lemma~\ref{thm:inner_exist}, when $ct_2$ is fixed, by solving the first equation we obtain $d$. Substituting it into the second equation and solving it gives us $\theta_1$. Then by Eq.~\eqref{eq:relation_d_theta} we obtain $\theta_2$.

We now show that when $N\geq 5$ and $c=10$, there exists solution $d$ to Eq.~\eqref{eq:first}, and $cd\pi$ is not an integer multiple of $2\pi$. For $\eta_0^0$ [see Eq.~\eqref{eq:set_size}] to be nonempty, we require $r<N-2$, and thus $N\geqslant 5$. 
In the first three cases of $N=5,6,7$, the solutions to Eq.~\eqref{eq:first} are $d\approx 0.30,0.38,0.42$, so $10d\pi$ is not an integer multiple of $2\pi$. 
When $N\geq 8$, the RHS of Eq.~\eqref{eq:first} varies in $(2,\frac{7}3 ]$. By the intermediate value theorem of continuous function, we only need to show that the LHS of Eq.~\eqref{eq:first} as a continuous function of $d\in(0.4,0.6)$ contains the closed interval $[2,\frac{7}3 ]$ in its codomain. Then there exists solution $d$ to Eq.~\eqref{eq:first}, and $cd\cdot\pi\in(4\pi,6\pi)$ is not an integer multiple of $2\pi$.
From $c=10$ we know that the LHS of Eq.~\eqref{eq:first} becomes:
\begin{equation*}
    h(d,x):=\frac{\cos(dx)-\cos(x)}{\cos(0.8x) -\cos(x)},
\end{equation*}
where $x:=\frac\pi{t_2}\in(0,\frac\pi 2)$. 
When $d=0.6$, we have $x-dx = 2(x-0.8x)$. As $\cos(x)$ is convex and monotone decreasing on interval $(0,\frac\pi 2)$, we have $h(0.6,x)<2$.
When $d=0.4$, from the following Lemma~\ref{lem:hx}, we know $h(0.4,x) > \lim_{x\to 0} h(0.4,x) = \frac{1-0.4^2}{1-0.8^2}=7/3$. 
Thus $h(d,x)$ as a function of $d\in(0.4,0.6)$, i.e. the LHS of Eq.~\eqref{eq:first}, contains $[2,\frac{7}3 ]$ in its codomain.

\begin{lemma}\label{lem:hx}
When $0<a<b<1,\, 0<x<\frac{\pi}{2}$, function $h(x)=\frac{\cos(ax)-\cos(x)}{\cos(bx) -\cos(x)}$ is monotone increasing, and $\lim_{x\to 0} h(x) = \frac{1-a^2}{1-b^2}$.
\end{lemma}
\begin{proof}
Denote the numerator and denominator of $h$ by $f$ and $g$. The numerator of $h'(x)$ is $z:=f'g-g'f$, and $z(0)=0$.
Substituting $f''=-a^2\cos(ax)+\cos(x)$ and $g''=-b^2\cos(bx)+\cos(x)$ into $z'=f''g-g''f$, we have $z'=(b^2-a^2)\cos(ax)\cos(bx) + (1-b^2)\cos(x)\cos(bx) + (a^2+1)\cos(ax)\cos(x) >0$. Hence, $z(x)>z(0)=0$ and therefore $h(x)$ is monotone increasing. By L'Hopital's rule, we know that $h(x)\to \frac{f''}{g''} = \frac{1-a^2}{1-b^2}$, when $x\to 0$.
\end{proof}

We now consider the limit of solution $d$ to Eq.~\eqref{eq:first}, when $N\to\infty$. 
First note that $x=\frac{\pi}{t_2} < \frac{2}{\sqrt{r}} \to 0$, as it's assumed that $t_2 \geq \frac{\pi}{2} \sqrt{r}$. From the Taylor expansion $\cos(x)=1-\frac{1}{2}x^2 + o(x^4)$, we know
\begin{equation}
    h(d,x) = \frac{1-d^2+o(x^2)}{1-(1-2/c)^2+o(x^2)} \to \frac{1-d^2}{1-(1-2/c)^2}.
\end{equation}
Combining $RHS\to 2$ when $N\to\infty$, we know the solution $d$ to Eq.~\eqref{eq:first} satisfies $\frac{1-d^2}{1-(1-2/c)^2} \to 2$. Hence, $d \to \sqrt{c^2-8(c-1)}/c$.

Finally we show that when $N\geq 5$, $t_2 \geq \frac{\pi}{2} \sqrt{r}$ and $d\in(0,1)$ (proved above), there exists solution $\theta_1$ to Eq.~\eqref{eq:second}.
Because $d\in(0,1)$, we have that the RHS $\in (0,1/\lambda_2 (1-\cos(\pi/t_2))$. The LHS is a continuous function of $\theta_1$, and its codomain contains the closed interval $[0,1]$. Thus by continuity, we only need to guarantee $1 \geq 1/\lambda_2 (1-\cos(\pi/t_2))$ then there exists solution $\theta_1$ to Eq.~\eqref{eq:second}.
When $N\geq 6$, it is easy to verify $\frac{r}2 > \frac{(N-r)(r+2)}{2(N-1)} = 1/\lambda_2$. So we only need to ensure that $1 \geq \frac{r}2 (1-\cos(\pi/t_2))$. 
As $1-\cos(x) \leqslant \frac{1}{2}x^2$, we only need to ensure that $1 > \frac{r \pi^2}{4t_2^2}$, which is
\begin{equation}
    t_2 \geq \frac{\pi}{2} \sqrt{r}.
\end{equation}
When $N=5$, we have $r=2$, hence $1/\lambda_2=9/8$. Thus we require $1 \geq \frac{9}{8}\frac{1}{2}\frac{\pi^2}{t_2^2}$, which is equivalent to $t_2 \geq \frac{3}4 \pi$. This is the same requirement as $t_2 \geq \frac{\pi}{2} \sqrt{r}$ since $t_2$ is an integer.
From the above two cases, we conclude that when $N\geq 5$, $d\in(0,1)$ and $t_2 \geq \frac{\pi}{2} \sqrt{r}$, there exists solution $\theta_1$ to Eq.~\eqref{eq:second}.

\end{appendices}

\end{document}